\documentclass[12pt, a4paper]{amsart}
\usepackage{fullpage}
\usepackage[utf8]{inputenc}
\usepackage{amssymb}
\usepackage{amsthm}
\usepackage{amsmath}

\usepackage[textsize=tiny]{todonotes}
\usepackage{color}
\definecolor{webcolor}{rgb}{0,0,1}
\definecolor{webbrown}{rgb}{.6,0,0}
\usepackage[
        colorlinks,
        linkcolor=webbrown, filecolor=webbrown,  citecolor=webbrown,
        pdfauthor={},
  pdftitle={},
]{hyperref}

\def\nilindex{E}
\def\nummax{F}

\DeclareMathOperator{\res}{res}
\DeclareMathOperator{\rres}{rres}
\DeclareMathOperator{\Rres}{Rres}

\def\den{\mathop{\mathrm{den}}\nolimits}

\newcommand{\Z}{\mathbf{Z}} 
\newcommand{\Q}{\mathbf{Q}}
\newcommand{\R}{\mathbf{R}} 
 
\newcommand{\F}{\mathbf{F}}

\newcommand{\MM}{\mathrm{M}}
\newcommand{\NN}{\mathrm{N}}

\theoremstyle{definition}
\newtheorem{theorem}{Theorem}[section]

\newtheorem{algorithm}[theorem]{Algorithm}
\newtheorem{definition}[theorem]{Definition}
\newtheorem{lemma}[theorem]{Lemma}
\newtheorem{proposition}[theorem]{Proposition}
\newtheorem{corollary}[theorem]{Corollary}

\newtheorem{example}[theorem]{Example}
\newtheorem{remark}[theorem]{Remark}

\newtheorem*{theorem*}{Theorem}
\newtheorem*{notation*}{Notation}

\newcommand{\Cont}{C}
\newcommand{\cont}{c}
\newcommand{\lc}{\mathrm{lc}}

\newcommand{\lcm}{\ensuremath{\mathrm{lcm}}}

\newcommand{\Mat}{\mathrm{Mat}}

\newcommand{\Ann}{\mathrm{Ann}}

\newtheorem{innercustomgeneric}{\customgenericname}
\providecommand{\customgenericname}{}
\newcommand{\newcustomtheorem}[2]{%
  \newenvironment{#1}[1]
      {%
        \renewcommand\customgenericname{#2}%
              \renewcommand\theinnercustomgeneric{##1}%
                 \innercustomgeneric
                   }
                     {\endinnercustomgeneric}
                     }

                     \newcustomtheorem{customproblem}{Problem}

\author{Claus Fieker}
\address{Claus Fieker\\
Fachbereich Mathematik\\
Technische Universität Kaiserslautern\\
67663 Kaiserslautern\\
Germany}
\email{fieker@mathematik.uni-kl.de}
\urladdr{http://www.mathematik.uni-kl.de/$\sim$fieker}

\author{Tommy Hofmann}
\address{Tommy Hofmann\\
  Fakultat für Mathematik und Informatik\\
  Universität des Saarlandes\\
  66123 Saarbrucken\\
  Germany}
\email{thofmann@mathematik.uni-kl.de}
\urladdr{http://www.mathematik.uni-kl.de/$\sim$thofmann}

\author{Carlo Sircana}
\address{Carlo Sircana\\
Fachbereich \ Mathematik\\
Technische \  Universität \  Kaiserslautern\\
67663 Kaiserslautern\\
Germany}
\email{sircana@mathematik.uni-kl.de}

\date{\today}

\makeatletter
\providecommand\@dotsep{5}
\def\listtodoname{List of Todos}
\def\listoftodos{\@starttoc{tdo}\listtodoname}
\makeatother

\title{Resultants over principal Artinian rings}

\begin{document}

\maketitle

\begin{abstract}
  The resultant of two univariate polynomials is an invariant of great importance in commutative algebra and vastly used in computer algebra systems.
  Here we present an algorithm to compute it over Artinian principal rings with a modified version of the Euclidean algorithm. Using the same strategy, we show how the reduced resultant and a pair of Bézout coefficient can be computed. Particular attention is devoted to the special case of $\Z/n\Z$, where we perform a detailed analysis of the asymptotic cost of the algorithm. Finally, we illustrate how the algorithms can be exploited to improve  ideal arithmetic in number fields and polynomial arithmetic over $p$-adic fields.
\end{abstract}

\section{Introduction}

The computation of the resultant of two univariate polynomials is an important task in computer algebra and it is used for various purposes in algebraic number theory and commutative algebra. It is well-known that, over an effective field $\mathbf F$, the resultant of two polynomials of degree at most $d$ can be computed in $O(\MM(d)\log d)$ (\cite[Section 11.2]{zurGathen2003}), where $\MM(d)$ is the number of operations required for the multiplication of polynomials of degree at most $d$. Whenever the coefficient ring is not a field (or an integral domain), the method to compute the resultant is given directly by the definition, via the determinant of the Sylvester matrix of the polynomials; thus the problem of determining the resultant reduces to a problem of linear algebra, which has a worse complexity.

In this paper, we focus on polynomials over principal Artinian rings, and we present an algorithm to compute the resultant of two polynomials. 
In particular, the method applies to polynomials over quotients of Dedekind domains, e.g. $\Z/n\Z$ or quotients of valuation rings in $p$-adic fields, and can be used in the computation of the minimum and the norm of an ideal in the ring of integers of a number field, provided we have a $2$-element presentation of the ideal.

The properties of principal Artinian rings play a crucial role in our algorithm: we recall them in Section 2 and we define the basic operations that we will need in the algorithm. The algorithm we present does not involve any matrix operations but performs only operations on the polynomials in order to get the result. Inspired by the sub-resultant algorithm, we define a division with remainder between polynomials in the case the divisor is primitive, which leads us immediately to an adjusted version of the Euclidean algorithm. As a result, we obtain algorithms to compute the reduced resultant, a pair of Bézout coefficient and the resultant of two polynomials. The runtime of the algorithms depends heavily on the properties of the base ring, in particular on the number of maximal ideals $F$ and on their multiplicity. The asymptotic cost of our method is $O(d\MM(d)\log(E)F)$, improving upon the asymptotic complexity of the direct approach, which consists of a computation of the echelon normal form of the Sylvester matrix.
In the special case of $\Z/n\Z$, we present a detailed analysis of the asymptotic cost of our method in terms of bit-operations, yielding a bit-complexity in
\[ O(d^2 \log(d) \log\log(d) \log(n)\log\log(n)). \]
Finally, we illustrate an algorithm to compute the greatest common divisor of polynomials over a $p$-adic field based on the ideas developed in the previous sections.
\subsection*{Acknowledgments}
The authors were supported by Project II.2 of SFB-TRR 195 `Symbolic Tools in
Mathematics and their Application'
of the German Research Foundation (DFG).

\section{Preliminaries}

Our aim is to describe asymptotically fast algorithms for the computation of (reduced) resultants
and Bézout coefficients. 
In this section we describe the complexity model that we use to quantify the runtime of the algorithms.

\subsection{Basic operations and complexity}\label{sec:complexity}

Let $R$ be a principal ideal ring.
By $\MM \colon \Z_{>0} \to \R_{>0}$ we denote multiplication time for $R[x]$, that is,
two polynomials in $R[x]$ of degree at most $n$ can be multiplied using at most $\MM(n)$ many additions and multiplications in $R$.
We will express the cost of our algorithms in terms of the number of \textit{basic operations} of $R$, by which we mean any of the following
operations:
\begin{enumerate}
\item
  Given $a, b \in R$, return $a + b$, $a - b$, $a \cdot b$ and true or false depending on whether $a = b$ or not.
\item
  Given $a, b \in R$, decide if $a \mid b$ and return $b/a$ in case it exists.
\item
  Given $a \in R$, return true or false depending on whether $a$ is a unit or not.
  \item
  Given $a, b \in R$, return $g, s, t \in R$ such that $(g) = (a, b)$ and $g = st + tb$.
\item
  Given  two ideals $I, J \subseteq R$, return a principal generator of the colon ideal $(I \colon J) = \{r \in R \mid rJ \subseteq I \}$.
\end{enumerate}

A common strategy of the algorithms we describe is the reduction to computations in non-trivial quotient rings.
The following remark shows that working in quotients is as efficient as working in the original ring $R$.

\begin{remark}
Let $r \in R$ and consider the quotient ring $\bar R = R/(r)$.
We will represent elements in $\bar R$ by a representative, that is, by an element of $R$.
It is straight foward to see that a basic operation in $\bar R$ can be done using at most $O(1)$ basic operations in $R$:
\begin{enumerate}
    \item Computing $\bar a + \bar b$, $\bar a - \bar b$, $\bar a \cdot \bar b$ is trivial. In order to decide if $\bar a = \bar b$, it is enough to test whether $a-b \in (r)$.
    \item  Deciding if $\bar a$ divides $\bar b$ in $R/(r)$ is equivalent to the computation of a principal generator $d$ of $(a, r) \subseteq R$ and testing if $d$ divides $b$.
    \item Deciding if an element $\bar a$ is a unit is equivalent to testing $(a, r) = R$.
    \item Given $\bar a$, $\bar b\in \bar R$, a principal generator $(\bar a, \bar b)$ is given by the image of a principal generator $g$ of $(a, b)$. If $sa + tb = g$, then $\bar s$ and $\bar t$ are such that $\bar g = \bar s \bar a + \bar t \bar b$.
    \item Given two ideals $\bar I, \bar J$ of $\bar R$, the colon ideal is generated by the principal generator of $((I+(r)) \colon (J+(r)))$. 
\end{enumerate}  
\end{remark}

\subsection{Principal Artinian rings}

The rings we will be most interested in are principal Artinian rings, that is, unitary commutative rings whose ideals are principal and satisfy the descending chain condition.
These have been studied extensively and their structure is well-known, see~\cite{Atiyah1969, McLean1973, BourbakiComm}.
Prominent examples of these rings include non-trivial quotient rings $\Z/n\Z$, $\F_q[x]/(f)$ and
non-trivial quotients of residually finite Dedekind domains.
By%
~\cite[chap. IV, \S 2.5, Corollaire 1]{BourbakiComm}, a principal Artinian ring $R$ has only finite many maximal ideals $\mathfrak m_1,\dotsc,\mathfrak m_r$ and there exist minimal positive integers $e_1, \dotsc ,e_r \in \Z_{>0}$ such that the canonical map $R \rightarrow \prod_{1 \leq i \leq r} R/\mathfrak m_i^{e_i}$ is an isomorphism of rings. We denote by $\pi_i \colon R \to R/\mathfrak m_i^{e_i}$ the canonical projection onto the $i$-th component.
For every index $1 \leq i \leq r$, the ring $R/\mathfrak m_i^{e_i}$ is a local principal Artinian ring.
We call $e_i$ the \textit{nilpotency index} of $\mathfrak m_i$ and denote by $E = \max_{1 \leq i \leq r} e_i$ the maximum of the nilpotency indices of the maximal ideals of $R$.
Note that $E$ is equal to the nilpotency index of the nilradical $\sqrt{(0)} = \{ r \in R \mid r \text{ is nilpotent}\}$ of $R$.
We will keep this notation for the rest of the paper whenever we work with a principal Artinian ring.

When investigating polynomials over $R$, the following well-known characterization of
nilpotent elements will be very helpful, see~\cite[chap. II, \S2.6, Proposition 13]{BourbakiComm}.

\begin{lemma}\label{lem:nilpotent}
  An element $a \in R$ is nilpotent if and only if $a \in \mathfrak m_1 \cap \dotsb \cap \mathfrak m_r = \mathfrak m_1 \dotsm \mathfrak m_r$.
\end{lemma}

As $R$ is a principal ideal ring, given elements $a_1, \dotsc,a_n$ of $R$, there exists an element $g$ generating the ideal $(a_1,\dotsc,a_n)$. We call $g$ a \textit{greatest common divisor} of $a_1,\dotsc,a_n$. Such an element is uniquely defined up to multiplication by units. By abuse of notation we will denote by $\gcd(a_1,\dotsc,a_r)$ any
such element.
Similarly, $\lcm(a_1,\dotsc,a_n)$ denotes a \textit{least common multiple} of $a_1,\dotsc,a_n$, that is, a generator of the intersection of the ideals generated by each $a_i$.
As $R$ is in general not a domain, quotients of elements are not well-defined.
To keep the notation lightweight, still for two elements $a, b \in R$ with $a \mid b$ we will denote by $b/a \in R$ an element $c \in R$ with $ca = b$ (the element $c$ is uniquely defined up to addition by an element of the annihilator $\Ann_R(a) = \{s \in R \mid sa = 0\}$).

The strategy for the (reduced) resultant algorithm will be to split the ring as soon as we encounter a non-favorable leading coefficient of a polynomial.
The (efficient) splitting of the ring is based on the following simple observation.

\begin{proposition}\label{prop:split}
Let $a \in R$ be a zero-divisor, which is not nilpotent.
Using $O(\log(E))$ many basic operations in $R$, we can find an idempotent element $e \in R$ such that
\begin{enumerate}
    \item 
      The canonical morphism $R \to R/(e) \times R/(1 - e)$ is an isomorphism with inverse $(\overline a, \overline b) \mapsto (1 - e)a + eb$.
    \item
      The image of $a$ in $R/(e)$ is nilpotent and the image of $a$ in $R/(1 - e)$ is invertible.
\end{enumerate}
\end{proposition}

\begin{proof}
For $i \in \Z_{\geq 0}$ consider the ideal $I_i = (a^i)$ of $R$. Since $R$ is Artinian, there exists $n \in \Z_{\geq 0}$ such that
$I_n = I_{n + 1}$. In particular $I_n^2 = I_n$, that is, $I_n$ is idempotent. (Note that we can always take $n = E$.)
Consider $c = a^n$. Since $c$ is a generator of the idempotent ideal $I_n$, we can find $b \in R$ such that $c^2 \cdot b = c$.
  Then $e = c \cdot b$ satisfies $e \in I_n$, $(1 - e)I_n = 0$, $e^2 = e$ and therefore (1) and (2) follow.
The cost of finding $e$ is the computation of the power of $a$, one multiplication and one division.
\end{proof}

\begin{definition}
  Let $a\in R$ be an element. We say that $a$ is a \textit{splitting element} if it is a non-nilpotent zero divisor.
\end{definition}

\section{Polynomials over principal Artinian rings}

We now discuss theoretical and practical aspects of polynomials over a principal Artinian ring $R$.
Note that due to the presence of zero-divisors the ring structure of $R[x]$ is more intricated
than in the case of integral domains.
For example, it is not longer true that every invertible element in $R[x]$ is a non-zero constant
or that every polynomial can be written as the product of its content and a primitive polynomial.
In this section, we show how to overcome these difficulties and describe asymptotically fast algorithms to compute
inverses, modular inverses and division with remainder (in case it exists).

\subsection{Basic properties}

We recall some theoretical properties of polynomials in $R[x]$. For the sake of completness, we include the short proofs.

\begin{definition}
 Let $f = \sum_{i = 0}^d a_i x^i \in R[x]$ be a polynomial.
We define the \textit{content ideal} $\Cont(f)$ of $f$ to be the ideal $(a_0,\dotsc,a_d)$ of $R$ generated by the coefficients of $f$. We say that $f$ is \textit{primitive} if $\Cont(f) = R$.
By abuse of notation, we will often denote by $\cont(f)$ a generator of this ideal.
\end{definition}

\begin{lemma}
  Let $f, g \in R[x]$ be primitive polynomials. Then the product $f g$ is primitive.
\end{lemma}

\begin{proof}
Assume that the product $fg$ is not primitive and let $C = C(fg)$ be its content ideal. Then $C$ is contained in
a maximal ideal $\mathfrak m$ of $R$, so $fg = 0 \bmod \mathfrak m [x]$. By assumption we have
$f$, $g\ne 0 \bmod \mathfrak m[x]$ yielding a contradiction since $(R/\mathfrak m)[x]$ is an integral domain.
\end{proof}

However, in general, due to the presence of idempotent elements, it is not true that if we write a polynomial $f \in R[x]$ as $f = \cont(f) \tilde f$ for some $f \in R[x]$, then $\tilde f$ is primitive.

\begin{example}
  Consider the non-primitive polynomial $f = 4x^2+8$ over $\Z/12\Z$, for which we clearly have $\Cont(f) = (4)$ and $\cont(f)= 4$. As $4 f = f$, we can set $\tilde f = f$ and write $f = c(f)f$.
\end{example}

Nevertheless, in case the content of $f$ is nilpotent, we can say something about the content of $\tilde f$.

\begin{lemma}\label{lem:content division}
  Let $f\in R[x]$ be a non-zero polynomial with nilpotent content. Let $\tilde f\in R[x]$ be any polynomial such that $f = c(f)\cdot \tilde f$. Then $c(\tilde f)$ is not nilpotent.
\end{lemma}

\begin{proof}
 Assume by contradiction that $\cont(\tilde f)$ is nilpotent. Now it holds that $\Cont(f) = \Cont(\cont(f)\cdot \tilde f) = \Cont(f) \Cont(\tilde f)$. Since $\Cont(f)$ and $\Cont(\tilde f)$ are nilpotent ideals, this is a contradiction.
\end{proof}

Next we give the well-known characterization of units and nilpotent elements of $R[x]$. We include a proof,
since it gives a bound on the degree of the inverse of an invertible polynomial.
Recall that $E$ is the nilpotency index of the nilradical of $R$. 

\begin{proposition}\label{prop:inverse_nilpotent}
  Let $f\in R[x]$ be a polynomial.
  \begin{enumerate}
      \item The polynomial $f$ is nilpotent if and only if its content is nilpotent.
      \item The polynomial $f$ is a unit if and only if the constant term $f_0$ of $f$ is a unit in $R$ and $f - f_0$ is nilpotent.
      \item If $f$ is invertible, then the degree of the inverse $f^{-1}$ is bounded by $\deg(f) \cdot E$.
  \end{enumerate}
\end{proposition}

\begin{proof}
(1): Assume that $\cont(f)$ is nilpotent. Since $\cont(f)$ divides $f$ in
  $R[x]$, also $f$ is nilpotent.  Vice versa, if $f$ is nilpotent, the
  projection of $f$ to a residue field $R/\mathfrak p$ is nilpotent too, so it
  must be zero. Hence all the coefficients of $f$ are in the intersection of
  all the maximal ideals of $R$, which coincides with the set of nilpotent
  elements by Lemma~\ref{lem:nilpotent}. As the content is then generated by
  nilpotent elements, it is nilpotent.

(2): Assume that the constant term $r$ is a unit of $R$ and that $g = f - r$ is
  nilpotent. Without loss of generality, $r =1$ since being a unit or nilpotent
  in $R[x]$ is invariant under multiplication with elements from $R^\times
  \subseteq R[x]^\times$.  Since $g^k = 0$ for a sufficiently large $k \in
  \Z_{\geq 1}$, we get $(1-g) \sum_{i=0}^{k-1} (-1)^ig^i = 1-g^k = 1$ showing
  that $f$ indeed is a unit. Vice versa, assume that $f$ is a unit. Then for
  every prime ideal $\mathfrak p$ of $R$ we have that the image of $f$ in
  $(R/\mathfrak p)[x]$ is a unit. In particular, as $R/\mathfrak p$ is a
  domain, all non-constant coefficients of $f$ are contained $\mathfrak p$.
  Since this holds for all prime ideals $\mathfrak p$, by
  Proposition~\ref{prop:inverse_nilpotent} the non-constant coefficients are
  nilpotent.

(3): If $g \in R[x]$ is nilpotent, then $g^E = 0$. Thus the claim follows as in
the proof of part~(2).
\end{proof}

\begin{remark}
The bound from Proposition~\ref{prop:inverse_nilpotent} is sharp. If $p \in \Z$ is a prime, then the inverse of the polynomial $1-px$ over $\Z/p^k\Z$ is $\sum_{i = 0}^{k-1} (px)^i$.
\end{remark}

Proposition~\ref{prop:inverse_nilpotent} allows us to use classical Hensel lifting (see \cite[Algorithm 9.3]{zurGathen2003}) to compute the inverse of units in $R[x]$. For the sake of completeness, we include the algorithm.

\begin{algorithm}\label{alg:inv}
  Given a unit $f\in R[x]$, the following steps return the inverse $f^{-1} \in R[x]$.
\begin{enumerate}
    \item Define $v_0$ as the inverse of the constant term of $f$.
    \item While $i < \log_2 (e \deg(f))$:
    \begin{enumerate}
    \item Set $v_{i+1} = v_i(2-v_if) \pmod{x^{2^i}}$.
    \item Increase $i$.
    \end{enumerate}
    \item Return $v_i$.
\end{enumerate}
\end{algorithm}

\begin{proposition}\label{prop:inv_algorithm}
  Algorithm~\ref{alg:inv} is correct and computes the inverse of $f \in R[x]$
  using $O(\MM(E \cdot \deg(f)))$ basic operations in $R$.
\end{proposition}

\begin{proof}
See \cite[Theorem 9.4]{zurGathen2003}.
\end{proof}

\subsection{Quotient and remainder}
We now consider the task of computing divisions with remainder.

\begin{remark}\label{remark:asym}
  Let $f, g \in R[x]$ be polynomials.
  If $g$ has invertible leading coefficient, one can use well-known asymptotically fast algorithms to find $q, r \in R[x]$ such that $f = qg + r$ and $\deg(r) < \deg(g)$ (see for example~\cite[Algorithm~9.5]{zurGathen2003}).
  This can be done using $O(\MM(d))$ basic operations in $R$.
\end{remark}

Things are more complicated when the leading coefficient of $g$ is not invertible.
Under certain hypotheses, we can factorize the polynomial as the product of a unit and a polynomial with invertible leading coefficient.

\begin{proposition}\label{prop:fun_factor}
Assume that $f = \sum_{i=0}^d f_i x^i \in R[x]$, is a primitive polynomial of degree $d$ and that there exists $0 \leq k \leq d$ such that for $k + 1 \leq i \leq d$ the coefficient $f_i$ is nilpotent and $f_k$ is invertible.
  Then there exists a unit $u \in R[x]^\times$ of degree $\deg(u) = d - k$ and a polynomial $\tilde f \in R[x]$ with invertible leading coefficient, $\deg(\tilde f) = k$ such that $f = u \cdot \tilde f$.
  The polynomials $\tilde f$ and $u$ can be computed using $O(\MM(d) \log(E))$ basic operations in $R$.
\end{proposition}

\begin{proof}
  This is just an application of Hensel lifting. More precisely,
  consider the ideal $\mathfrak a = (f_i \mid i = k+1, \dotsc,d)$ of $R$.
  Since $\mathfrak a$ is generated by nilpotent elements, it is nilpotent and $\mathfrak a^\nilindex
  = \{0\}$.
  Consider the polynomials $\bar u = 1 + \sum_{i = 1}^{d-k}
  f_{k + i}x^i$ and $\bar f = \sum_{i=0}^k f_ix^i$ in $R[x]$. By construction $f \equiv
  \bar u \cdot \bar f \pmod{\mathfrak a}$ and $\bar u \equiv 1
  \pmod{\mathfrak a}$. Thus $\bar u$, $\bar f$ are coprime modulo
  $\mathfrak a$ and $1 \equiv 1 \cdot \bar u + 0 \cdot \bar f \pmod{\mathfrak a}$. Furthermore, the leading coefficient of $\bar f$ is
  invertible. Therefore, by means of Hensel lifting and
  since $\mathfrak a$ is nilpotent, we can lift the factorization of $f$ modulo $\mathfrak a$ to factorization of $f$ in $R[x]$.
  The lifting can be done using~\cite[Algorithm 15.10]{zurGathen2003}. As in our case $f$ is not monic, the degree of the lift of $\bar u$ will increase during the lifting process, but since the polynomial $\bar f$ has invertible leading coefficient, the degree of $u$ will be $d - k$. The cost of every step in the lifting process is $O(\MM(d))$, as it involves a constant number of additions, multiplication and divisions between polynomials of degree at most $d$. As the number of steps we need is at most $\log(\nilindex)$, the claim follows. 
\end{proof}

\begin{example}
  In $\Z/8\Z[x]$, $f = 2x^5+x^3+1$ satisfies the hypotheses of Proposition~\ref{prop:fun_factor}. The corresponding factorization is $f = (2x^2+4x+1)\cdot ( x^3+6x^2+4x+1 ) $.
\end{example}

\begin{proposition}\label{prop:divrem}
  Let $f, g \in R[x]$ be polynomials of degree at most $d$ and assume that $g$ is primitive. Then using at most $O(\MM(\nilindex \cdot d)\cdot \min(F, d))$ basic operations in $R$
  we can find $q, r \in R[x]$ such that $f = qg + r$ and $0 \leq \deg(r) < \deg(g)$, where $\nummax$ is the number of maximal ideals of $R$.
\end{proposition}

\begin{proof}
  Assume first that $g$ satisfies the hypotheses of Proposition~\ref{prop:fun_factor}. Then we can compute a factorization $g = u\tilde g$ with $\tilde g \in R[x]$ monic of degree $\leq d$ and $u \in R[x]^\times$ a unit. As $\tilde g$ is monic, we can perform division with remainder of $f$ by $\tilde g$ and find $q, r \in R[x]$ such that $f = q \tilde g +r$ and $\deg(r) < \deg(\tilde g) \leq \deg(g)$. Multiplying $q$ by the inverse of $u$, we get $f = qu^{-1}g+r$.
  By Remark~\ref{remark:asym} and Proposition~\ref{prop:inv_algorithm}, the division needs $O(\MM(d))$ and the inversion $O(\MM(\nilindex \cdot d))$ basic operations respectively. As the degree of $u^{-1}$ is bounded by $\deg(u) \nilindex \leq \deg(g) \nilindex$ and the degree $q$ by $\deg(f) = d$ , the final multiplication of $u^{-1}$ with $q$ 
  requires $O(\MM(\nilindex \cdot d))$ basic operations.
  Thus the costs are in $O(\MM(\nilindex \cdot d))$.

  Now, we deal with the general case. In particular $g$ has trivial content but it does not satisfy the assumption of Proposition~\ref{prop:fun_factor}. This means that the first non-nilpotent coefficient $c$ of $g$ is a splitting element.
  Therefore, by Proposition~\ref{prop:split}, using $O(\log \nilindex)$ basic operations we can find an isomorphism $R \to R_1 \times R_2$ of $R$ with the direct product of two non-trivial quotient rings.
  In the quotients, the coefficient $c$ will be either nilpotent or invertible by Proposition~\ref{prop:split}.
  If $c$ is invertible in the quotient $R_i$, then the projection of $g$ to $R_i[x]$ satisfies the assumption of Proposition~\ref{prop:fun_factor}. Thus the division can be performed using $O(\MM(\nilindex \cdot d))$ basic operations.
  In case $c$ is nilpotent in the quotient $R_i$, we need to repeat this process until
  the first non-nilpotent coefficient of the polynomial will be invertible. This has to happen eventually, as the content of the polynomial is trivial and the number of maximal ideals is finite.
  As at every step the degree of the term of the polynomial which is non-nilpotent decreases, the splitting can happen at most $\min(\nummax, d)$ times.
  At the end, we reconstruct the quotient and the remainder by means of the Chinese remainder theorem.
  It follows that the algorithm requires $O(\MM(\nilindex \cdot d)\cdot \min(d, \nummax))$
  basic operations in $R$.
\end{proof}

\subsection{Modular inverses}
 
Finally note that using a similar strategy as in Algorithm~\ref{alg:inv} we can compute modular inverses.

\begin{algorithm}\label{alg:invmod}
  Given a unit $u\in R[x]$ and a polynomial $f\in R[x]$ with invertible leading coefficient, the following steps return the inverse $u^{-1} \in R[x]$ modulo $f$.
\begin{enumerate}
    \item Define $v_0$ as the inverse of the constant term of $u$.
    \item While $i < \log_2 (\deg(f))$:
    \begin{enumerate}
        \item Set $v_{i+1} = v_i(2-v_iu) \pmod{x^{2^i}}$.
        \item Increase $i$.
    \end{enumerate}
    \item While $i < \log_2 (\nilindex\cdot \deg(u))$:
    \begin{enumerate}
        \item Set $v_{i+1} = v_i(2-v_iu) \pmod{f}$.
        \item Increase $i$.
    \end{enumerate}
    \item Return $v_{i}$.
\end{enumerate}
\end{algorithm}

\begin{lemma}
  Algorithm~\ref{alg:invmod} is correct and computes the inverse of $u$ modulo $f$ using $O(\MM(\deg(f))\log (\nilindex\deg(u)))$ basic operations in $R$.
\end{lemma}  
\begin{proof}
  The correctness follows from \cite[Theorem 9.4]{zurGathen2003} as above. The complexity result follows from the fact that the degrees of the polynomials that we compute during the algorithm are bounded by $\deg(f)$.
\end{proof}

\section{Resultants and reduced resultants via linear algebra}\label{sec:resultant}

In this section, we will describe algorithms to compute the resultant, reduced resultant and the Bézout coefficients of univariate polynomials over an arbitrary principal ideal ring $R$, which in this section is not assumed to be Artinian.
The algorithms we present here will be based on linear algebra over $R$,
for which the complexity is described in~\cite{Storjohann2000}.
Note that in~\cite{Storjohann2000} a slightly different notions of basic
operations is used, which can be used to derive the basic operations from Section~\ref{sec:complexity}.
For the sake of simpiclity in this section we will use the term basic operations to refer to basic operations as described
in \cite{Storjohann2000}.

We start by recalling the definition of the objects that we want to compute.
Let $f, g \in R[x]$ be polynomials of degree of $n$ and $m$ respectively.
Recall that the Sylvester matrix of the pair $f, g$ is the matrix $S(f, g) \in \Mat_{(n + m) \times (n + m)}(R)$ representing the $R$-linear map
\[ \varphi \colon P_m \times P_n \to P_{n + m},\, (s, t) \mapsto sf + tg ,\]
where $P_k = \{ h \in R[x] \mid \deg(h) < k\}$, with respect to the canonical basis $(x^k, x^{k - 1}, \dotsc, x, 1)$.

\begin{definition}
  Let $f, g \in R[x]$ be polynomials. We define the $\textit{resultant}$ $\res(f, g)$ of $f, g$ to be the determinant $\det(S(f, g)) \in R$
  of the Sylvester matrix, 
  and the \textit{reduced resultant} $\Rres(f, g)$ of $f, g$ as the ideal $(f,g)\cap R$.
  Two elements $u, v \in R[x]$ are called \textit{Bézout coefficients of the reduced resultant} of $f$ and $g$, if they satisfy $uf + vg \in R$ and $(uf + vg) = \Rres(f, g)$.
  As usual, by abuse of notation, we will call any generator of $\Rres(f, g)$ a reduced resultant of $f, g$ and denote it by $\rres(f, g)$.
\end{definition}

\subsection{Reduced resultant.} We begin by showing that, similar to the resultant, also the reduced resultant can be characterized in terms of invariants of the Sylvester matrix (at least in the case that one of the leading coefficients is invertible).

\begin{lemma}\label{lem:deg_cofactor}
Let $f$, $g$, $u$, $v\in R[x]$ such that $uf+vg = r \in R$ and assume that the leading coefficient of $f$ or $g$ is invertible.
Then we can find $\tilde v$, $\tilde u\in R[x]$ such that
  $\deg(\tilde u)<\deg(g)$, $\deg(\tilde v)<\deg(f)$ and $r = \tilde u f+\tilde v g$.
\end{lemma}

\begin{proof}
Without loss of generality, we may assume that $f$ is monic.
Thus we can use polynomial division to write
  $v = q f + \tilde v$ and $\deg(\tilde v) < \deg(f)$. Now
$uf + vg = uf + (qf +\tilde v)g = (u+qg)f + \tilde v g$.
Let $\tilde u = u+qg$. Then we have
  $\deg(g) + \deg(\tilde v) < \deg(g) + \deg(f)$ and, since $f$ is monic, 
  $\deg(\tilde u f) = \deg(\tilde u) + \deg(f) = \deg(\tilde vg)$. This
  shows $\deg(\tilde u)<\deg(g)$ as claimed.
\end{proof}

Recall that the strong echelon form of a matrix over a principal ideal ring is the same as the Howell form with
reordered rows. In case of a principal ideal domain, the strong echelon form is the same as a Hermite normal form,
where the rows are reordered, such that all the pivot entries are on the diagonal, see~(\cite{Howell1986, Storjohann2000, Fieker2014}).
In case the matrix has full rank, it is just the last diagonal entry.
We will make use only of the following property of the upper right strong echelon form of a matrix $A \in \Mat_{k \times k}(R)$:
If $v = (v_1,\dotsc,v_k) \in R^k$ is contained in the row span of $A$ and $v_1 = v_2 = \dotsb = v_l = 0$ for some $1 \leq l \leq k$,
then $v$ is in the row span of the $k - l$ rows. In particular, if $v_1 = v_2 = \dotsb = v_{k - 1} = 0$, then $v$ is a multiple
of the last row of $A$, that is, $v_k$ is a multiple of the last diagonal entry of $A$.

\begin{proposition}\label{reduced resultant via Howell}
  Let $f, g \in R[x]$ be polynomials of degree $n, m$ respectively, $k = n + m$ and $H = (h_{ij})_{1 \leq i,j \leq k} \in \Mat_{k \times k}(R)$ the upper right strong echelon form of $S(f, g)$. Assume that one of the leading coefficients of $f$ and $g$ is invertible. Then $(h_{k, k}) = \Rres(f, g)$.
\end{proposition}

\begin{proof}
Under the $R$-isomorphism $R^k \to P_k, (v_{k-1},\dotsc,v_0) \mapsto \sum_{0\leq i < k} v_i x^i$,
the row span of $S(f, g)$ is mapped onto $\{ uf + vg \mid \deg(u) < m, \deg(v) < n\}$.
The statement now follows from Lemma~\ref{lem:deg_cofactor} and the properties of the strong echelon form.
\end{proof}

\begin{corollary}
  If $R$ is a domain, $f, g \in R[x]$, the matrix $S(f, g)$ non-singular and one of the leading coefficients of $f, g$ invertible, then $\Rres(f, g)$ is generated by the last diagonal entry of the upper right Hermite normal form of $S(f, g)$.
\end{corollary}

\begin{remark}\hfill
\begin{enumerate}
    \item 
    Proposition~\ref{reduced resultant via Howell} is in general not correct if both leading coefficients are not invertible, 
    since cofactors satisfying the degree conditions given in Lemma~\ref{lem:deg_cofactor} may not exists.
    For example, let $p \in \Z$ be a prime and consider $f = px+1$ and $g = px+1+p$ in $(\Z/p^2\Z)[x]$.
    As $f$ is invertible by Proposition~\ref{prop:inverse_nilpotent}, the ideal $(f, g)$ is $(\Z/p^2\Z)[x]$ and the reduced resultant $\rres(f, g)$ is one.
    However, there are no constants $a, b\in \Z/p^2\Z$ such that $af + bg = 1$.
    The Sylvester matrix $S(f, g)$ is equal to
    \[ \begin{pmatrix} p & 1 \\ p & 1 + p\end{pmatrix} \]
    and has Howell and strong echelon form equal to
    \[ \begin{pmatrix} p & 1 \\ 0 & p \end{pmatrix}. \]
    Thus 
    \[ \{ uf + vg \mid u, v \in (\Z/p^2\Z)[x], \deg(u) < 1, \deg(v) < 1\} \cap \Z/p^2\Z = (p), \]
    while 
    \[ \{ uf + vg \mid u, v \in (\Z/p^2\Z)[x]\} \cap \Z/p^2\Z = \Rres(f, g) = (1). \]
\item
  A similar result as in Proposition~\ref{reduced resultant via Howell} can be found in~\cite{zurGathen1998}.
  There the authors show that in case $R$ is a discrete valuation ring, the reduced resultant of $f$, $g \in R[x]$ is equal
  to the largest elementary divisor of $S(f, g)$ (see~Lemma~3.8 of op. cit.).
\end{enumerate}
\end{remark}

\begin{corollary}\label{cor:rres:linalg}
  Let $f, g \in R[x]$ be polynomials of degree $n, m$ respectively.
  Assume that one of the leading coefficients of $f$ and $g$ is invertible.
  Both a reduced resultant $\rres(f, g)$ and Bézout coefficients of a reduced resultant of $f, g$ can be computed using $O((n + m)^\omega)$
  many basic operations in $R$, where $\omega$ is the exponent of matrix multiplication.
\end{corollary}

\begin{proof}
Follows from Proposition~\ref{reduced resultant via Howell} and~\cite[Chapter 4]{Storjohann2000} (see also~\cite{Storjohann1998}).
\end{proof}

\subsubsection*{Resultants.}
For the sake of completeness, we also state the corresponding result for the resultant. 
Note that here, in contrast to the reduced resultant, we do not need any assumption on the leading coefficients of the polynomials.
While the resultant can easily be computed as the determinant of the Sylvester matrix, a pair of Bézout coefficients can be found via linear algebra.

\begin{proposition}
Let $f, g \in R[x]$ be two polynomials of degree $n$, $m$ respectively.
Both the resultant $\res(f, g)$ and Bézout coefficients for the
resultant can be computed using $O((n + m)^\omega)$ many basic operations in $R$.
\end{proposition}

\begin{proof}
  By \cite[Theorem 3]{Collins1971}, there exist $v$, $u\in R[x]$ such that
  $\deg(u)<\deg(g)$, $\deg(v)<\deg(f)$ and $\res(f, g) =  u f+ v g$. Therefore the computation of a pair of Bézout coefficients reduces to a determinant computation and linear system solving involving a $k \times k$ matrix over $R$.
  Thus the result follows from \cite[Chapter 4]{Storjohann2000}.
\end{proof}

\section{Resultants and polynomial arithmetic}

In this section, we show how to compute the resultant, reduced resultant and Bézout coefficients of univariate polynomials over a principal Artinian ring by directly manipulating the polynomials, avoiding
the reduction to linear algebra problems.
At the same time, this will allow us to get rid of the assumption on the leading coefficients of the polynomials, as present in
Corollary~\ref{cor:rres:linalg}.
For the rest of this section we will denote by $R$ a principal Artinian ring.

\subsection{Reduced resultant}

Let $f, g \in R[x]$ be polynomials. The basic idea of the reduced resultant algorithm is to use that $(f,g) = (f - pg, g)$ for every $p \in R[x]$ in order to make the degree of the operands decrease. Thus the computation reduces to the following base cases:

\begin{lemma}\label{lem:base}
  Let $c, d \in R \subseteq R[x]$ be constant polynomials and $f \in R[x]$ a polynomial with invertible leading coefficient and $\deg(f) > 0$.
  Then the following hold:
  \begin{enumerate}
    \item
      $\Rres(c, d) = (c, d)$ and $\rres(c, d) = \gcd(c, d)$,
    \item
      $\Rres(f, c) = (c)$ and $\rres(f, c) = c$.
  \end{enumerate}
\end{lemma}

\begin{proof}
  (1): Clear.
  (2): Since $c$ is contained in $(f, c) \cap R$, we can reduce to the computation of $(\overline{f})\cap R/(c)$, where $\overline f \in (R/(c))[x]$ is the projection of $f$ modulo $(c)$. As $f$ has invertible leading coefficient, $(\overline f)\cap R/(c) = (0)$.
\end{proof}

Unfortunately, this process may fail, mainly for the following reasons: the leading coefficient of the divisor is a zero divisor or the polyomials are not primitive. We now describe a strategy to overcome these issues.

\subsubsection*{Reduction to primitive polynomials}
First of all, we show how to reduce to the case of primitive polynomials. Let
$f$, $g\in R[x]$ be polynomials of positive degree. We want to show that either
we find a splitting element $r \in R$ or we reduce the computation to primitive
polynomials. In some case, we will need to change the ring over which we are
working. To clarify this, when necessary, we will specify the ring in which we
are computing the reduced resultant by writing it as a subscript, for example,
$\rres_R$ denotes the reduced resultant over $R$.
For a quotient $R/(r)$ we will denote by $\rres_{R/(r)}(f, g) \in R$ a lift of the reduced resultant of the polynomials
$\bar f, \bar g \in (R/(r))[x]$.

\begin{lemma}\label{lem:content_monic_nilpotent}
  Let $f$, $g\in R[x]$ be polynomials and assume $g$ is of positive degree with
  invertible leading coefficient. Let $h \in R[x]$ such that $f = c(f)
  \cdot h$. Then $\rres_R(f, g) = c(f) \cdot \rres_{R/\Ann(c(f))}(h,
  g)$.
\end{lemma}
\begin{proof}
  Let $r = \rres(f, g)$ be the reduced resultant of $f$ and $g$ and write $s\cdot c(f)h +tg = r$ with $s, t\in R[x]$. Since $tg \equiv r \pmod {c(f)}$, the polynomial $g$ has invertible leading coefficient and $r$ is a constant, both $t$ and $r$ lie in the ideal generated by $c(f)$. Indeed, the additivity of the degree of the product holds if one of the polynomials has invertible leading coefficient. Thus 
\[
  sh +\frac{t}{c(f)}g \equiv \frac{r}{c(f)} \bmod {\Ann(c(f))}
\]
giving the result.
\end{proof}

\begin{lemma}\label{lem:primitive_nilpotent}
  Let $f$, $g\in R[x]$ be two polynomials of positive degree.
  Assume that $g$ is primitive and write $f = c(f)h$, where $h \in R[x]$. Suppose that
      neither $c(f)$, $c(h)$ nor any of the coefficients of $g$ are splitting elements.
      Denoting by $g = u \cdot \tilde g$ a factorization from Proposition~\ref{prop:fun_factor}, we have
      $\rres(f, g) = c(f)\cdot \rres_{R/\Ann(c(f))}( h, \tilde g)$.
\end{lemma}

\begin{proof}
As all the coefficients of $g$ are either invertible or nilpotent,
$g$ satisfies the hypotheses of Proposition~\ref{prop:fun_factor}, so that indeed $g = u \cdot \tilde g$, where $u$ is a unit and $\tilde g$ is monic. As $u$ is a unit, $(f, g) = (f, \tilde g)$ and the same holds for the reduced resultant. Therefore, we can apply Lemma~\ref{lem:content_monic_nilpotent} and this translates immediately to the result.
\end{proof}

\begin{lemma}\label{lem:both_nilpotent}
  Let $f, g \in R[x]$ be two nilpotent polynomials. Let $d$ be a generator of
  $(c(f), c(g))$ and let $\tilde f$, $\tilde g\in R[x]$ be polynomials such that
  $f = d\cdot \tilde f$ and $g = d\cdot \tilde g$. Then $\rres(f, g) = d\cdot
  \rres_{R/\Ann(d)}(\tilde f, \tilde g)$ and either $\tilde f$ or $\tilde g$ is
  not nilpotent.
\end{lemma}

\begin{proof}
As $(f, g) = d\cdot (\tilde f, \tilde g)$, the relation between the reduced resultants holds.
We now prove that either $\tilde f$ or $\tilde g$ is not nilpotent. Let $f_1, g_1\in R[x]$ be polynomials such that $f = d\cdot c(\tilde f) \cdot f_1$ and $g = d\cdot c(\tilde g) \cdot g_1$.
Looking at the content, we obtain $c(f) = d\cdot c(\tilde f)\cdot c(f_1)$ and $c(g) = d \cdot c(\tilde g) \cdot c(g_1)$. Therefore, $d = (c(f), c(g)) = d\cdot (c(\tilde f)\cdot c(f_1), c(\tilde g) \cdot c(g_1))$. If both $\tilde f$, $\tilde g$ are nilpotent, we therefore get a contradiction.
\end{proof}

We use these three cases either to split the ring $R$ or to reduce to the case of primitive polynomials.

\begin{algorithm}\label{alg:ppa}
Input: $f$, $g\in R[x]$ non-constant.\\
Output: $r \in R$ or $(c, \bar f, \bar g) \in R \times R[x] \times R[x]$.
\begin{enumerate}
\item If $f$ and $g$ are primitive, return $(1, f, g)$.
\item If $c(f)$ is a splitting element, return $c(f)$.
\item If $c(g)$ is a splitting element, return $c(g)$.
\item If $c(f)$ and $c(g)$ are nilpotent, then:
    \begin{enumerate}
        \item Compute $d = (c(f), c(g))$.
        \item Compute $\tilde f = f/d$, $\tilde g = g/d$.
        \item Apply Algorithm~\ref{alg:ppa} to $(\tilde f, \tilde g)$. If it returns an element $r \in R$, return $r$.
              Otherwise if it returns $a_1, f_1, g_1$ then return $(d\cdot a_1, f_1, g_1)$.
    \end{enumerate}
\item If $c(f)$ is not nilpotent, swap $f$ and $g$.
\item Find the term $a_i x^i$ of $g$ of highest degree which is not nilpotent.
\item If $a_i$ is invertible, then:
\begin{enumerate}
    \item Factorize $g = u\cdot \tilde g$ using Proposition~\ref{prop:fun_factor} and set $\tilde f = f/c(f)$.
    \item Apply Algorithm~\ref{alg:ppa} to $(\tilde f, \tilde g)$. If it returns an element $r \in R$, return $r$.
          Otherwise, if it returns $a_1, f_2, g_2$, then return $(c(f)\cdot a_1, f_2, g_2)$.
\end{enumerate}
\item If $a_i$ is not invertible, return $(a_i, f, g)$.
\end{enumerate}
\end{algorithm}

\begin{proposition}\label{prop:reduction_primpolys}
  Let $f$, $g\in R[x]$ be two polynomials of degree at most $d$.
  Then Algorithm~\ref{alg:ppa} either returns $(c, \tilde g, \tilde f)$ with $c \in R$ and $\tilde f, \tilde g \in R[x]$ primitive polynomials such that $\rres_R(f, g) = c \cdot \rres_{R/\Ann(c)}(\tilde f, \tilde g)$ or returns a splitting element of $R$. 
  The algorithm requires $O(\MM(d) \log(E))$ basic operations.
\end{proposition}

\begin{proof}
  If both polynomials are primitive, the statement is trivial.
  Assume now that one of the polynomials is primitive, say $g$.
  If the content of $f$ is a zero divisor, the algorithm is clearly correct. Otherwise, the algorithm follows the proof of Lemma~\ref{lem:primitive_nilpotent} and performs a recursive call on the resulting polynomials. By Lemma~\ref{lem:content division}, in the recursive call, either both polynomials are primitive or the content of one of them is a splitting element, as desired.
  
  Let us now assume that both polynomials are not primitive.
  If one of them is not nilpotent, its content is a splitting element of $R$.
  Therefore the only case that remains is when both polynomials are nilpotent. By
  means of Lemma~\ref{lem:both_nilpotent}, we can reduce to the case when at
  least one of the two polynomials is not nilpotent, and we have already dealt
  with this case above. Therefore the claim follows.
  
 We now analyse the runtime.
All the operations except for the factorization in Step~(7a) using Proposition~\ref{prop:fun_factor} are at most linear in $d$. If we are in the case of the factorization of Proposition~\ref{prop:fun_factor}, then in the recursive call we have a monic polynomial and a polynomial which is not nilpotent. Therefore the recursive call will take at most linear time in $d$ and the algorithm requires in this case $O(\MM(d)\log(E))$ operations.
\end{proof}

\subsubsection*{Chinese remainder theorem}
Algorithm~\ref{alg:ppa} returns in some cases a splitting element $r \in R$. In such a case, we use Proposition~\ref{prop:split} and we continue the computation over the factor rings. 
Therefore we need to explain how to recover the reduced resultant over $R$ from the one computed over the quotients.

\begin{lemma}\label{lem:split_res}
  Let $f, g \in R[x]$ be two polynomials and let $e$ be a non-trivial idempotent.
  Denote by $\pi_1$ and $\pi_2$ the projections onto the components $(R/(e))[x]$ and $R/(1-e))[x]$ respectively. Then $\pi_i(\Rres(f, g)) = \Rres(\pi_i(f), \pi_i(g))$ for $i = 1, 2$.
  In particular we have $\rres(f, g) = e \cdot \rres_{R/(1 - e)}(f, g) + (1 - e) \rres_{R/(e)}(f, g)$.
\end{lemma}

\begin{proof}
Let $r$ be the reduced resultant of $f, g$; as $r\in(f, g)$, there exist $s, t\in R[x]$ such that $sf+tg = r$. Therefore, applying $\pi_i$ we get $\pi_i(r) \in (\pi_i(f), \pi_i(g))$ and therefore $\pi_i(\Rres(f, g))\subseteq \Rres(\pi_i(f), \pi_i(g))$. On the other hand, let $r_i = \rres(\pi_i(f), \pi_i(g)) = u_i\pi_i(f) + v_i\pi_i(g)$. Then the Chinese remainder theorem implies that there exists $r \in R$ and $u, v \in R[x]$ such that $\pi_i(r) = r_i$ and $uf + vg = r$, as desired.
\end{proof}

\subsubsection*{The main algorithm}
Using Algorithm~\ref{alg:ppa}, we may assume that the input polynomials for the reduced resultant algorithm are primitive.
In order to compute the reduced resultant of two polynomials $f$ and $g$, we want to perform a modified version of the Euclidean algorithm on $f$, $g$.
During the computation, we will potentially split the base ring using Proposition~\ref{prop:split} and reconstruct the result using Lemma~\ref{lem:split_res}. 
We will now describe the computation of the reduced resultant.
Before stating the algorithm, we briefly outline the basic idea. Let us assume that $\deg(f) \geq \deg(g)$.
\begin{itemize}
    \item If $g$ has invertible leading coefficient, we can divide $f$ by $g$ with the standard algorithm (Remark~\ref{remark:asym}).
    \item If the leading coefficient of $g$ is not invertible, we want to apply Proposition~\ref{prop:fun_factor}. If the first non-nilpotent coefficient is invertible, then we get a factorization $g = u \cdot \tilde g$ where $u$ is a unit and $g$ is monic. Therefore $\rres(f, g) = \rres(f, \tilde g)$ and $(f, \tilde g)$ satisfy the the hypotheses of item~(1).
    If it is not invertible, then it is a splitting element and Proposition~\ref{prop:split} applies.
\end{itemize} 
We repeat this until the degree of one of the polynomials drops to $0$. In this case, we can just use one of the base cases from Lemma~\ref{lem:base}.

Summarizing, we get the following recursive algorithm to compute the reduced resultant.

\begin{algorithm}[Reduced resultant]\label{alg:rres}
Input: Polynomials $f$, $g$ in $R[x]$.\\
Output: A reduced resultant $\rres(f, g)$.
\begin{enumerate}
\item If $f$ or $g$ is constant, use Lemma~\ref{lem:base} to return $\rres(f, g)$.
\item If $\deg(g) > \deg(f)$, swap $f$ and $g$.
\item Apply Algorithm~\ref{alg:ppa} to $(f, g)$.
\item If Step~(2) returns an element $r$ (which is necessarily a splitting element), then:
\begin{enumerate}
    \item Compute a non-trivial idempotent $e \in R$ using Proposition~\ref{prop:split}.
    \item Recursively compute $r_1 = \rres_{R/(e)}(f, g)$, $r_2 = \rres_{R/(1 - e)}(f, g)$ and return
          $er_2 + (1-e)r_1$.
\end{enumerate}
\item Now Step~(2) returned $(c, f_1, g_1)$.
\item If $c \not\in R^\times$, then return $\rres(f, g) = c \cdot \rres_{R/\Ann(c)}(f_1, g_1)$.
\item Now $c \in R^\times$, so that $\rres(f, g) = c \cdot \rres(f_1, g_1)$ and both $f_1$ and $g_1$ are primitive.
\item Let $a_i x^i$ be the term of $g_1$ of highest degree which is not nilpotent.
\item If $a_i$ is invertible, then
\begin{enumerate}
    \item Factorize $g_1 = u\cdot \tilde g_1$ with $u \in R[x]^\times$ a unit and $g_1 \in R[x]$ monic using Proposition~\ref{prop:fun_factor}.
    \item Divide $f_1$ by $g_1$ using Remark~\ref{remark:asym} to obtain $q, r \in R[x]$ with $\deg(r) < \deg(g_1)$ and $f_1 = qg_1 +r$.
    \item Return $c\cdot \rres(\tilde g_1, r)$
\end{enumerate}
\item If $a_i$ is not invertible, then $a_i$ is a splitting element and we proceed as in Step~(3) with $(f, g)$ replaced by $(f_1, g_1)$ and multiply the result by $c$.
\end{enumerate}
\end{algorithm}

\begin{theorem}\label{thm:rres}
  Algorithm~\ref{alg:rres} is correct and terminates. 
  If the degree of $f$ and $g$ is bounded by $d$, then the number of basic operations is in $O(d\MM(d)\nummax\log(E))$.
\end{theorem}

\begin{proof}
We first discuss the correctness of the algorithm. At every recursive call we either have that:
\begin{itemize}
    \item the polynomials we produce generate the same ideal as the starting ones; in this case correctness is clear;
    \item the reduced resultant of the input polynomials is the is the same as the reduced resultant of the polynomials we produce in output up to a constant, as stated in Proposition~\ref{prop:reduction_primpolys}; 
    \item we find a splitting element and continue the computation over the residue rings, following Lemma~\ref{lem:split_res}.
\end{itemize}
Termination is straightforward too, as at every recursive call we either split the ring, pass to a residue ring or the sum of the degrees of the polynomials decreases and these operations can happen only a finite number of times.
Let us analyse the complexity of the algorithm.
Algorithm~\ref{alg:ppa} costs at most $O(\MM(d)\cdot \log(E))$ operations by Proposition~\ref{prop:reduction_primpolys} and it is the most expensive operation that can be performed in every recursive call. The splitting of the ring can happen at most $\nummax$ times and every time we need to continue the computation in every quotient ring. The recursive call in Step~$6$ will start again with Step~$8$, as the input polynomials are primitive. Therefore, as passing to the quotient ring takes a constant number of operations, it does not affect the asymptotic cost of the algorithm. The recursive call in Step~$(9)$ can happen at most $d$ time. Summing up, the total cost of the algorithm is $O(d\MM(d)\nummax\log(E))$ operations.
\end{proof}

\begin{example}\label{example:rres1}
  We consider the polynomials $f = x^2+2x+3$ and $g = x^2+1$ over $\Z/12\Z$. The polynomials are primitive and monic, so we go directly to step $9$ of the algorithm and we divide $f$ by $g$, so that we get $(f, g) = (g, 2x+2)$. As the second polynomial has now content $2$, we can use it as a splitting element. This means that we need to continue the computation over $\Z/3\Z$ and $\Z/4\Z$. In the first of these rings, as $2$ is invertible, we divide $g$ by $2x+2$, obtaining the ideal $(2x+2, 2)$. Thus $\rres_{\Z/3\Z}(2x+2, g) = 1$. Let us now consider the second ring, $\Z/4\Z$. Here $2$ is nilpotent, so we get $\rres_{\Z/4\Z}(g, 2x+2) = 2\cdot \rres_{\Z/2\Z}(g, x+1)$ as $g$ is monic. By dividing, $(g, x+1)= (0, x+1)$ and therefore $\rres_{\Z/2\Z}(g, x+1) = 0$. Therefore, $\rres_{\Z/4\Z}(g, 2x+2) = 0$.
  Applying the Chinese remainder theorem, we therefore get $\rres(f, g) = (4)$.
\end{example}

\begin{remark}\label{rem:alwaysmonic}
When one of the input polynomials has invertible leading coefficient, at every recursive call of Algorithm~\ref{alg:rres} one of the two polynomials will still have invertible leading coefficient.
\end{remark}

\subsubsection*{Computation of Bézout coefficients}\label{subsect:bezout_coeffs}

Let $f, g \in R[x]$ be polynomials. We want find two polynomials $a, b\in R[x]$ such that $af + bg = \rres(f, g)$.
To this end, in the same way as in the Euclidean algorithm, we will keep track of the operations that we perform during the computation of the reduced resultant.
In order to describe the algorithm, we just need to explain how to obtain cofactors in the base case and how to update cofactors during the various operations of Algorithm~\ref{alg:rres}.
We begin with the base case, which follows trivially from Lemma~\ref{lem:base}.

\begin{lemma}
Let $f, g \in R[x]$ be polynomials.
\begin{enumerate}
    \item 
  If $f$ and $g$ are constant, let $a, b \in R$ such that $(f, g) = (af + bg)$ as ideals of $R$. Then $\rres(f, g) = af + bg$.
  \item
    If $f$ has positive degree and $g$ is constant, then $\rres(f, g) = g = 0\cdot f + 1 \cdot g$.
  \end{enumerate}
  In particular, if $f$ or $g$ is constant, then Bézout coefficients for the reduced resultant can be computed using $O(1)$ basic operations.
\end{lemma}

An easy calculation shows the following:

\begin{lemma}
Let $f, g \in R[x]$ be polynomials of positive degree.
\begin{enumerate}
    \item If $f = qg +r$ with $q, r\in R[x]$ and $a_0, b_0 \in R[x]$ satisfy $a_0g + b_0r = \rres(g, r)$, then $\rres(f, g) = (a_0 - qb_0)f + b_0 g$.

    \item If $c \in R, f  = c \cdot \tilde f$, $\tilde f \in R[x]$ and $a_0, b_0 \in R[x]$ satisfy $\rres(\tilde f, g) = a_0 f + b_0 g$, then $\rres(f, g) = a_0 f + cb_0 g$.
    \item
      If $g = u\cdot \tilde g$ where $u \in R[x]^\times$ is a unit, $\tilde g \in R[x]$ and $a_0, b_0 \in R[x]$ satisfy $\rres(\tilde g, f) = a_0\tilde g + b_0f = \rres(f, g)$, then $\rres(f, g) = (a_0 \cdot u^{-1})\cdot f + b_0 g$.
\end{enumerate}
In particular, given $a,b \in R[x]$ of degree at most $d$, in cases~(1) and~(2) we can compute Bézout coefficients using at most $O(\MM(d))$ basic operations. In case~(3) we can compute Bézout coefficients using $O(\MM(d\cdot E))$ basic operations.
\end{lemma}

Finally we mention the case where the ring is split using a non-trivial idempotent.

\begin{lemma}
    Assume that $e$ is a non-trivial idempotent, $f, g \in R[x]$ and $\rres(\bar f, \bar g) = a_1 \bar f + b_1 \bar g$ over $R/(e)$ with $a_1,b_1 \in (R/(e))[x]$ and $\rres(\bar f, \bar g) = a_2 \bar f + b_2 \bar g$ over $R/(1 - e)$ with $a_2,b_2 \in (R/(1 - e))[x]$. 
    Then $\rres(f, g) = ((1 - e)a_1 + ea_2)f + ((1-e)b_1 + eb_2)g$.
    The complexity is bounded by $O(d)$, where $d$ is the maximum of the degrees of $a_1$, $b_1$, $a_2$, $b_2$.
\end{lemma}
    
Notice that, as the inverse of a unit may have large degree, the Bézout coefficient that we compute this way may have large degree too.
If we assume that one of the polynomials has invertible leading coefficient, we can reduce the degrees of the cofactors using Lemma~\ref{lem:deg_cofactor}. Under this assumption, this reduction must be done every time we invert the unit in order to control the degrees of the cofactors.
More precisely, assume that during one of the steps of the algorithm $f$ has invertible leading coefficient and we factorize the second polynomials $g$ as $u\cdot \tilde g$ and $\tilde g$ is monic. Then we do not need to compute the full inverse of $u$, but to update the cofactors we only need to know the inverse of $u$ modulo $f$.
The complexity analysis changes completely under this assumption, so let us first assume that one of the input polynomials has invertible leading coefficient.

\begin{lemma}\label{lem:bezout_monic_case}
  Let $f, g\in R[x]$ be polynomials of degree at most $d$ and assume that $f$ or $g$ have invertible leading coefficient. Then we can compute Bézout coefficients $a, b \in R[x]$ with $af + bg = \rres(f, g)$ using $O(d\MM(d)\nummax\log (d \nilindex))$ basic operations.
\end{lemma}

\begin{proof}
  By Remark~\ref{rem:alwaysmonic} we know that at every recursive call one of the polynomials will have invertible leading coefficient. Thus the most expensive operation that we perform in Algorithm~\ref{alg:rres} is given by the computation of the inverse of the unit modulo the polynomial with invertible leading coefficient in Step 9, which requires at most $O(\MM(d)\log (d\nilindex))$ operations.
  Therefore the claim follows as in Theorem~\ref{thm:rres}.
\end{proof}

Now, we determine the complexity of the algorithm in the case none of the input polynomials has invertible leading coefficient.

\begin{lemma}
  Let $f, g\in R[x]$ be polynomials of degree at most $d$. Then we can compute Bézout coefficients $a, b \in R[x]$ with $af + bg = \rres(f, g)$ using at most $O(\MM(d\cdot \nilindex) + d\MM(d)\nummax\log d\nilindex)$ operations.
\end{lemma}
\begin{proof}
We can assume that both polynomials are primitive. If one of them has invertible leading coefficient, Lemma~\ref{lem:bezout_monic_case} applies.
Otherwise, by means of Proposition~\ref{prop:fun_factor} and, if needed, splitting the ring, we can reduce to the case of a monic polynomial $\tilde g$.
Therefore we can compute Bézout coefficients for $f$ and $\tilde g$ using $O(d\MM(d)\nummax\log d\nilindex )$ by Lemma~\ref{lem:bezout_monic_case}. To get Bézout coefficients for $f$ and $g$, we need to invert the unit, which requires at most $O(\MM(d\cdot \nilindex))$ basic operations by Proposition~\ref{prop:inv_algorithm}.
\end{proof}

\subsection{Resultant}
Let $f, g \in R[x]$ be two polynomials. In this subsection, we show how to compute the resultant $\res(f, g)$.
We want to follow the same approach as in the computation of the reduced resultant:
Via a modified version of the Euclidean algorithm,
we compute successive remainders in order to reduce the degree of the polynomials, until we arrive at one of the following base cases, which are easily verified using the definition.

\begin{lemma}\label{lem:base_res}
Let $f\in R[x]$ be a polynomial of positive degree and $a, b, c\in R$ constants. Then the following hold:
\begin{enumerate}
    \item $\res(f, c) = c^{\deg(f)}$,
    \item $\res(x-a, x-b) = a-b$.
\end{enumerate}
\end{lemma}

In order to apply this strategy, we need to understand how the resultant behaves with respect to division with remainder.
We have the following classical statement.

\begin{lemma}\label{lem:resultant_division}
 Let $f, g\in R[x]$ be non-constant polynomials. If $f = qg + r$ with
 $\deg(q) = \deg(f) - \deg(g)$ and $\deg(r) < \deg(g)$, then
 $\res(f, g) = \lc(g)^d \res(r, g)$, where $d = \deg(f) - \deg(r)$.
\end{lemma}

As usual, the problem is that division works well only if the leading coefficient of the quotient is invertible. Therefore, we will now explain how to reduce to this case.
Given polynomials $f, g\in R[x]$, we will describe an algorithm that either returns a splitting element $r\in R$ or returns two primitive polynomials $f_1, g_1$ and a constant $c$ such that $c\cdot \res(f_1, g_1) = \res(f, g)$. The resultant behaves well with respect to the content of a polynomial, which is a consequence of the multilinearity of the resultant.

\begin{lemma}
 Let $f, g \in R[x]$ be polynomials and $c\in R$ be a constant. Then $\res(f, cg) = c^{\deg(f)}\res(f, g)$.
\end{lemma}

\subsubsection*{Chinese remainder theorem}
Every time we encounter a splitting element, we can use Proposition~\ref{prop:split} to split $R$. In order to exploit this in the computation of the resultant, we need to understand the behaviour of the resultant with respect to this operation.
As the projection onto one of the components is a homomorphism of rings, we recall the behaviour of the resultant with respect to homomorphisms, see~\cite[Theorem 9.2]{Geddes1992}.

\begin{lemma}\label{prop:resultant_homo}
 Let $f, g\in R[x]$ be non-constant polynomials and let $\varphi \colon R \rightarrow S$ be a ring homomorphism. 
 \begin{itemize}
    \item If $\deg(f) = \deg(\varphi(f))$ and $\deg(g) > \deg(\varphi(g))$, then 
    \[
    \varphi(\res(f, g)) =\varphi(\lc(f))^{\deg(g)-\deg(\varphi(g))} \res(\varphi(f), \varphi(g));
    \]
    \item if $\deg(f) > \deg(\varphi(f))$ and $\deg(g) > \deg(\varphi(g))$, then $\varphi(\res(f, g)) = 0$.
\end{itemize}
\end{lemma}

\begin{proof}
The image of the resultant $\varphi(\res(f, g))$ is equal to the determinant
of $\varphi(S(f, g))$, i.e. the matrix whose entries are the images of the
entries of $S(f, g)$ under $\varphi$. If the degrees of the polynomials are
invariant under $\varphi$, then $\varphi(S(f, g)) = S(\varphi(f), \varphi(g))$,
giving the first formula. The second result follows by using Laplace expansion
on the first $\deg(g)-\deg(\varphi(g))$ rows of $\varphi(S(f, g))$. The third
statement follows from the fact that, under the given hypotheses, the first
column of the matrix $\varphi(S(f, g))$ is zero. 
\end{proof}

\begin{proposition}
Let $\varphi \colon R \to R_1 \times R_2$ be a ring isomorphism, $\pi_i \colon R \to R_i$ the canonical projections, $i = 1, 2$, and $f, g \in R[x]$ polynomials.
For $i = 1,2$ let $f_i = \pi_i(f)$ and $g_i = \pi_i(g)$ and set
\[ r_i = \begin{cases} 0, \quad &\text{if } \deg(f_i) < \deg(f) \text{ and } \deg(g_i) < \deg(g), \\ 
                    \lc(g)^{\deg(f) - \deg(f_i)}\res(f_i, g_i), \quad &\text{if } \deg(f_i) \leq \deg(f) \text{ and } \deg(g_i) = \deg(g), \\
                    \lc(f)^{\deg(g) - \deg(g_i)}\res(f_i, g_i), \quad &\text{if } \deg(f_i) = \deg(f) \text{ and } \deg(g_i) \leq \deg(g). \\
\end{cases} \]
Then $\res(f, g) = \varphi^{-1}(r_1, r_2)$.
\end{proposition}

\begin{proof}
Follows from Lemma~\ref{prop:resultant_homo}.
\end{proof}

\begin{corollary}\label{cor:res_crt}
Let $e$ be a non-trivial idempotent and $f, g \in R[x]$ polynomials of degree at most $n$. Given the resultants of $f, g$ over $R/(e)$ and $R/(1 - e)$ respectively, we can compute the resultant $\res(f, g)$ using $O(\log(n))$ many basic operations.
\end{corollary}

It remains to describe what happens when we factorize one of the polynomials, say $g$, as $g = u \cdot \tilde g$ with $u \in R[x]^\times$ a unit and $\tilde g \in R[x]$ monic.

\begin{lemma}
 Let $f, g, h\in R[x]$ be non-constant polynomials. If $\deg(f) + \deg(g) = \deg(fg)$, then
 $\res(fg, h) =\res(f, h)\res(g, h)$.
\end{lemma}

\begin{proof}
The formula holds for every integral domain. In particular, it holds over a
multivariate polynomial ring over $\Z$. Let $a_0, \dots, a_{\deg(f)}$ be the
coefficients of $f$, $b_0, \dots, b_{\deg(g)}$ the coefficients of
$g$ and $c_0, \dots, c_{\deg(h)}$ the coefficients of $h$. Consider the polynomial
ring $S = \Z[s_{f, 0}, \cdots , s_{f, \deg(f)}, s_{g, 0},\cdots,  s_{g, \deg(g)}, s_{h, 0},
\cdots, s_{h, \deg(h)}]$ and the homomorphism $\varphi \colon S \rightarrow R$
sending $s_{f,i}$ to $a_i$, $s_{g, i}$ to $b_i$ and $s_{h, i}$ to $c_i$.
The map $\varphi$ induces a homomorphism $S[x] \to R[x]$, which by abuse of notation we also denote by $\varphi$.
By construction, $f, g, h$ are in the image of $\varphi$ and hence there exist $\tilde f$, $\tilde g$ and $\tilde h$ in $S[x]$ which map to $f, g$ and $h$ respectively. 
Invoking Proposition~\ref{prop:resultant_homo} twice, we have
    \begin{align*}
         \res(fg, h) =  \res(\varphi(\tilde f \tilde g), \varphi(\tilde g)) = \varphi(\res(\tilde f \tilde g, \tilde h)) &= \varphi(\res(\tilde f, \tilde h))\varphi(\res(\tilde g, \tilde h)) \\
         &= \res(f, h)\res(g, h). \qedhere
    \end{align*}
\end{proof}

By virtue of this lemma, if we factorize $g$ using Proposition~\ref{prop:fun_factor} as $g = u \cdot \tilde g$, where $u \in R[x]^\times$ is a unit and $\tilde g \in R[x]$ is monic, then $\res(f, g) = \res(f, u)\cdot \res(f, \tilde g)$. We continue with the Euclidean algorithm in the computation of $\res(f, \tilde g)$. For what concerns $\res(f, u)$, we pass to the reciprocal polynomial.

\begin{definition}
  Let $f\in R[x]$ be a polynomial. We define the reciprocal polynomial $i(f) \in R[x]$ as the polynomial $x^{\deg(f)}f(1/x)$.
\end{definition}

\begin{lemma}
 Let $f, g \in R[x]$ and assume that the constant term of $f$ is non-zero. Then $\res(f, g) = (-1)^{(\deg(f)\cdot \deg(g))}\cdot \lc(i(f))^{\deg(g) - \deg(i(g))}\res(i(f), i(g))$.
\end{lemma}
\begin{proof}
 Note that $f(0) \neq 0$ is equivalent to $\deg(i(f)) = \deg(f)$. If $\deg(i(g)) = \deg(g)$, the Sylvester matrix of the reversed polynomials $i(f)$, $i(g)$ is obtained by swapping the columns of the Sylvester matrix of $f$, $g$, so the claim follows.
 If $\deg(i(g)) <= \deg(g)$, we use Laplace formula for the determinant on the first columns in order to reduce to the Sylvester matrix of $i(f)$ and $i(g)$; this correspond to the multiplication by a power of leading term of $i(f)$.
\end{proof}

\subsubsection*{The resultant algorithm}
\begin{algorithm}[Resultant]\label{alg:resultant}
Input: Polynomials $f$, $g \in R[x]$.\\
Output: The resultant $\res(f, g)$.
\begin{enumerate}
\item If $f$ or $g$ is constant or $\deg(f) = \deg(g) = 1$, use Lemma~\ref{lem:base_res} to return $\rres(f, g)$.
\item If $\deg(g) > \deg(f)$, swap $f$ and $g$.
\item Write $f = c(f) \cdot \tilde f$ and set $r_1 = c(\tilde f)$.%
\item If $r_1 \not\in R^\times$, then:
\begin{enumerate}
  \item Split the ring $R \simeq R/(r_1) \times R/(1 - r_1)$ using $r_1$ as a splitting element.
  \item Compute $\res_{R/(r_1)}(\tilde f, g)$ and $\res_{R/(1 - r_1)}(\tilde f, g)$ recursively.
  \item Use Corollary~\ref{cor:res_crt} to determine $\res(\tilde f, g)$ and return $c(f) \cdot \res(\tilde f, g)$.
\end{enumerate}
\item Write $g = c(g) \cdot \tilde g$ and set $r_2 = c(\tilde g)$.
\item If $r_2\not\in R^\times$, then
\begin{itemize}
  \item Split the ring $R \simeq R/(r_2)\times R/(1 - r_2)$ using $r_2$ as a splitting element.
  \item Compute $\res_{R/(r_2)}(\tilde f, \tilde g)$ and $\res_{R/(1 - r_2)}(\tilde f, \tilde g)$ recursively.
  \item Use Corollary~\ref{cor:res_crt} to determine $\res(\tilde f, \tilde g)$ and return $c(f) \cdot c(g) \cdot  \res(\tilde f, \tilde g)$.
\end{itemize}
\item Let $b_ix^i$ be the term of $\tilde g$ of highest degree with $b_i$ not nilpotent.
\item If $b_i$ is not invertible, then
\begin{itemize}
    \item Split the ring $R \simeq R/(b_i)\times R/(1 - b_i)$ using $b_i$ as a splitting element.
    \item Compute $\res_{R/(b_i)}(\tilde f, \tilde g)$ and $\res_{R/(1 - b_i)}(\tilde f, \tilde g)$ recursively.
  \item Use Corollary~\ref{cor:res_crt} to determine $\res(\tilde f, \tilde g)$ and return $c(f) \cdot c(g) \cdot \res(\tilde f, \tilde g)$.
\end{itemize}
\item Factorize $\tilde g = u \hat g$ where $u \in R[x]^\times$ is a unit and $\hat g \in R[x]$ is monic using Proposition~\ref{prop:fun_factor}. 
\item Compute the remainder $f_1$ of the division of $\tilde f$ by $\hat g$.
\item Set $f_2 = i(\tilde f)$ and $u_1 = i(u)$.
\item Compute the remainder $f_3$ of the division of $f_2$ by $u_1$.
\item Return $c(f)\cdot c(g) \cdot \lc(u_1)^{\deg(f_2) - \deg(f_3)} \cdot \res(f_3, u_1) \cdot \res(f_1, \hat g)$
\end{enumerate}
\end{algorithm}

\begin{theorem}\label{thm:res}
 Algorithm~\ref{alg:resultant} is correct and for input of degree at
 most $d$ requires at most $O(d\MM(d)\log(\nilindex)\nummax)$ basic operations.
\end{theorem}

\begin{proof}
The algorithm terminates, as the number of maximal ideals and of quotients of $R$ is finite and every time we perform a division, the sum of the degrees of the polynomials decreases. Correctness follows from the discussion above.
Let us now compute the computational cost of the algorithm.

We denote by $R(d)$ the cost of the algorithm for input polynomials of degree at most $d$. Notice that $R(d)$ is superlinear, that is $\limsup_{d \to \infty}R(d)/d = \infty$, as at every recursive call we compute a division of polynomials and this operation is already superlinear. This implies that $R(d) \geq R(d-s) + R(s)$ for $d \gg 0$ and all $0\leq s < d$.
Steps~(1) and~(3) require $O(d)$ basic operations.
The division in Step~(10) takes at most $O(\MM(d))$ basic operations, while the computation of
the factorization in Step~(9) requires $O(\MM(d)\log \nilindex)$ basic operations.
Every time we use Proposition~\ref{prop:fun_factor}, we need to split the computation
into two resultants of lower degree. In particular, the degrees of the input polynomials in the recursive call are bounded by $d-s-1$ and $s$, where $s$ is the degree of the unit $u$ in the factorization at Step~(9).
Putting together all these considerations, we get the bound
\[
  R(d) \leq \MM(d)\log(\nilindex) + R(d-1)
\]
in the case we never split the ring. Whenever we split the ring, then we need to
compute the resultant in both quotient rings and apply Corollary~\ref{cor:res_crt}.
This can happen at most $\nummax$ times. Therefore, the asymptotic complexity of
the algorithm is $O(d\MM(d)\log(\nilindex)\nummax)$.
\end{proof}

\begin{example}
  We want to compute the resultant of $f = x^3+2x+1$ and $g = x^3+2x^2+2$ over $\Z/4\Z$. As the polynomials are primitive and monic, we divide directly $f$ by $g$. Denoting by $f_1$ the remainder $f_1 = 2x^2+2x+3$ so that $\res(f, g) = \res(f_1, g)$. Now, both polynomials are primitive but the leading coefficient of $f_1$ is not invertible. Therefore, we apply Proposition~\ref{prop:fun_factor}. As $f_1$ is a unit, the corresponding factorization is trivial. Following the algorithm, we pass to the reverse polynomials, so that $\res(f, g) = \res(i(f_1), i(g))$. As $i(f_1)$ has invertible leading coefficient, we can divide $i(g)$ by $i(f_1)$, so that $\res(f, g) = \res(i(f_1), 2x+1)$. Both polynomials are primitive and the algorithm recognises that $2x+1$ is a unit and passes again to the reciprocal polynomials, $\res(f, g) = \res(f_1, x+2)$. Dividing again, $\res(f, g) = \res(3, x+2)$ so that we are in the case of Lemma~\ref{lem:base_res}. Therefore $\res(f, g) = 3$.
\end{example}

\subsubsection*{Resultant ideal}
Let $f, g \in R[x]$ be polynomials.
In a lot of applications, we are interested in the ideal $(\res(f, g)) \subseteq R$ generated by the resultant
and not in the resultant itself. In such a case, the algorithm can be simplified,
as the resultant of a unit and a monic polynomial is a unit of $R$:

\begin{lemma}\label{lem:res_unit}
Let $f, u\in R[x]$ be polynomials. Assume that $u \in R[x]^\times$ is a unit and $f$ has invertible leading coefficient. Then $\res(u, f)$ is a unit.
\end{lemma}

\begin{proof}
 Let $\mathfrak m$ be a maximal ideal of $R$. By Lemma~\ref{lem:nilpotent}, the projection of $u$ modulo $\mathfrak m$ is a non-zero constant. This implies that the Sylvester matrix $S(u, f)$ is invertible over the residue field $R/\mathfrak m$. Applying the same argument to all maximal ideals of $R$, we obtain $S(f, g) \in R^\times$.
\end{proof}

Therefore, if we make sure that at every recursive call the polynomial of higher degree has invertible leading coefficient, in the recursive call at Step~(13) we only need to compute one resultant instead of two, since
\[\lc(u_1)^{\deg(f_2) - \deg(f_3)} \res(f_3, u_1) = \lc(u_1)^{deg(f_3) - \deg(f_3)} \res(f_2, u_1) = \res(\tilde f, u) \in R^\times \]
by Lemma~\ref{lem:res_unit}.
If the polynomial of higher degree does not have invertible leading coefficient, we still need to invoke Proposition~\ref{prop:fun_factor}.
Even if this does not change the asymptotic complexity, it makes the computation faster in practice. 

\subsection{Multivariate polynomials}
A classical application of the resultant algorithm is the computation of the resultant of
multivariate polynomials  with respect to one variable using a modular approach, see
\cite{Collins1971} and \cite{Geddes1992} for the case $R = \Z$.
Here, for the sake of simplicity we will focus on bivariate polynomials. Of course, using an inductive process it is
possible to use the same argument to deal with more variables.
The idea of the algorithm is to reduce to the univariate case in order to use Algorithm~\ref{alg:resultant}.
  
Let $f, g \in R[x, y]$ be polynomials and assume that we want to compute the resultant with respect to $y$,
so that $S(f, g)$ is a matrix with entries in $R[x]$. The degree of the resultant is bounded by
$B = \deg_y(g)\deg_x(f)+\deg_y(f)\deg_x(g)$. This bound can be easily computed just by looking at the degrees
of the entries of $S(f, g)$, see also \cite[(9.34)]{Geddes1992}.

Assume that there exist $B + 1$ interpolation points in $R$, that is, there are elements $a_0, \dots, a_{B} \in R$ such that $a_i-a_j \in R^\times$ is invertible for $i \neq j$.
We then consider the ideals $I_i = (x- a_i)\subseteq R[x, y]$, $0 \leq i \leq B$, whose residue rings are all isomorphic to $R[y]$.
As $a_i-a_j$ is invertible, the ideals are pairwise coprime. Thus we can apply the Chinese remainder theorem,
compute the resultants over the residue rings by means of Algorithm~\ref{alg:resultant} and
reconstruct it in $R[x]$ using the formula given by Proposition~\ref{prop:resultant_homo},
in the same way we did for the resultant itself.

\begin{algorithm}[Resultant of bivariate polynomials]\label{alg:biv_resultant}\*
\newline Input: Polynomials $f$, $g \in R[x, y]$ and $a_0 \dotsc, a_B \in R$ such that $a_i-a_j\in R^\times$, where $b = \deg_y(g)\deg_x(f) + \deg_y(f)\deg_x(g)$.\\
Output: The resultant $\res_y(f, g) \in R[x]$.
\begin{enumerate}
  \item For $i = 0, \dotsc, B$, compute the resultant $r_i \in R$ of $f, g$ in $R[x, y]/(x-a_i) \cong R[y]$ using Algorithm~\ref{alg:resultant}.
  \item Using Proposition~\ref{prop:resultant_homo}, reconstruct the resultant over $R[x]$.
\end{enumerate}
\end{algorithm}

\begin{remark}
In general it is not easy (or not possible at all) to find such elements.
We show what to do in the case of $\Z/n\Z$ and for residue rings of valuations rings of a $p$-adic field.
\begin{enumerate}
\item
Assume first that $R = \Z/n\Z$. Let $p_1, \dots, p_k$ be the prime numbers smaller than $B$.
Then we factorize $n$ as $n = p_1^{e_1}\dots p_k^{e_k} m$ with $m$ coprime to every $p_i$. Over $\Z/m\Z$, the elements $0, 1, \dots, B$ are a set satisfying the requirements and therefore we can compute the resultant by using Algorithm~\ref{alg:biv_resultant}.
To treat the rings $\Z/p_i^{e_i}\Z$, we construct a base ring extension as follows.
Denote now by $p$ be one of the small prime divisors $p_i$ of $n$ and $e = e_i$ the corresponding exponent.
Let $k = \log_p(B)$, $S = \Z/p^e\Z$ and $\lambda \in S[t]$ a monic polynomial of degree $k$, which is irreducible modulo $p$.
Then in $S[t]/(\lambda)$ the elements $a_{i,j} = i \bar t^j$ for $i \in \{0, \dots, p-1\}$ and $j \in \{0, \dots, k\}$ are a set of elements such that $a_{i,j}- a_{i_1,j_1}$ is invertible if $(i, j) \neq (i_1, j_1)$.
Therefore we can work over the extension ring $S[t]/(\lambda)$ and compute the resultant over $(S[t]/(\lambda))[x]$ by using Algorithm~\ref{alg:biv_resultant}. As the input polynomials have coefficients in the subring $\Z/p^e\Z \subseteq S[t]/(\lambda)$, the result will be in $\Z/p^e\Z$, so that the reconstruction is not affected by this.
\item
Let us now assume that $R$ is the valuation ring of a $p$-adic field with residue field $\kappa$.
Let $q$ the cardinality of the residue field. If $q > B$, then we can take lifts of the elements of the residue field $\kappa$ as elements $a_0, \dotsc, a_B$ and apply Algorithm~\ref{alg:biv_resultant}.
Otherwise, as we did in the first case, we need to work over a ring extension. Let $k = \log_q(B)$ and let $ \lambda \in R[t]$ be a polynomial of degree $k$ which is irreducible over $\kappa$.
Then in $S = R[t]/(\lambda)$ we can consider the elements $b_{i,j} = a_i\bar t^j$, for $i\in\{1, \dotsc, q\}$ and $j \in \{0, \dotsc, k-1\}$, where the $a_i\in R$ are lifts of the elements of $\kappa$. The difference $b_{i,j} - b_{i_1, j_1}$ is invertible if $(i, j) \neq (i_1, j_1)$, so that these element can be used to apply Algorithm~\ref{alg:biv_resultant}.
\end{enumerate}
\end{remark}

\section{Complexity over residue rings of the integers}
In Section~5 we have determined the algebraic complexity of computing resultant and reduced resultants over a principal ideal ring $R$.
The number of operations depend not only on the degree of the input polynomials, but also on the number $F$ of maximal ideals and the nilpotency index $E$.
For the ring $R = \Z/n\Z$, we want to translate the algebraic complexity into a bit complexity statement.
We will denote by $\NN \colon \Z_{\geq 0} \to \Z_{\geq 0}$ the bit complexity of
multiplying to two integers. Thus two integers $a, b \in \Z$ with
$\lvert a \rvert, \lvert b \rvert \leq 2^k$ can be multiplied using $O(\NN(k))$ many bit operations. 
In this regard, \cite[Theorem 1.5]{Storjohann2000} shows that in $\Z/n\Z$,
\begin{itemize}
    \item basic operations (1), (2) and (3) can be computed using $O(\NN(\log(n)))$ many bit operations, and
    \item basic operations (4) and (5) can be computed using $O(\NN(\log(n)) \log(n))$ many bit operations.
\end{itemize}

We will improve upon the naive translation into bit complexity of Theorems~\ref{thm:res} and~\ref{thm:rres} by exploiting the fact that
the operation in non-trivial quotients of $\Z/n\Z$ are less expensive than operations in $\Z/n\Z$.

\begin{lemma}
Let $f\in \Z/n\Z[x]$ be a primitive polynomial. The bit complexity of the lifting of Proposition~\ref{prop:fun_factor} is $O(\MM(d)\NN(\log(n)))$.
\end{lemma}
\begin{proof}
Without loss of generality, we can assume $n = m^{2^k}$.
The cost of every step of the lifting is $O(\MM(d)\NN(\log m^{2^i}))$, if we are working over $\Z/m^{2^i}\Z$ . This means that the complexity of the entire algorithm is $\MM(d) \sum_{i = 1}^k \NN(\log m^{2^i})$.
As $\NN$ is superlinear, we get
\begin{align*}
 \sum_{i = 1}^k \NN(\log m^{2^i}) & \leq \NN(\log m^{2^k}) + \NN(\sum_{i = 1}^{k-1} \log m^{2^i} ) \\
     & = \NN(\log m^{2^k}) + \NN(\log m\sum_{i = 1}^{k-1} 2^i) \leq  2\NN(\log m^{2^k})
\end{align*}
Summarizing, the cost of the Hensel lifting is $O(\MM(d)\NN(\log n))$, as desired.
\end{proof}

\begin{theorem}
 Let $f, g \in \Z/n\Z[x]$ be polynomials of degree at most $d$. The bit complexity of the reduced resultant and resultant algorithms is $O(d\MM(d)\NN(\log(n)))$.
\end{theorem}

\begin{proof}
  We only need to show that the number of distinct prime factors of $n$ does not affect the complexity.
  We show the result by induction on the number of distinct factors of $n$.
  If $n$ is prime, it follows from the general statement. Let us now assume that the statement is true for every $n$ with $k$ factors and let $n$ be a modulus with $k+1$ distinct prime factors.
  Let us assume that, in the algorithm, we find a factor $m$ of $n$ such that $m$ and $n/m$ are coprime. Then, following the algorithm, we need to split the computation into two, i.e. continue over $\Z/m\Z$ and $\Z/(n/m)\Z$. By induction, we know that over these rings the number of operations is bounded by $O(d\MM(d)\NN(\log(m)))$ and $O(d\MM(d)\NN(\log(n/m)))$. As
  $\NN(\log(m)) + \NN(\log(n/m)) \leq \NN(\log(m) + \log(n/m)) = \NN(\log(n))$ by the superlinearity of $\NN$, the statement follows.
\end{proof}

\section{Polynomials over complete discrete valuation rings}
We now show how to use the algorithms of the previous sections to improve operations over complete discrete valuation rings. The most important example is given by the valuation rings of local fields, for example, the ring $\Z_p$ of $p$-adic integers.
To this end let $R$ be a complete discrete valuation ring with fraction field $K$.
Let $\pi$ be a uniformizer of $K$, that is, $\pi$ is a generator of the unique maximal ideal of $R$.
In particular, if we fix a set of representatives $S \subseteq R$ of $R/(\pi)$, then every element $a \in R$ can be written as $a = \sum_{i=0}^\infty a_i \pi^i$ for unique elements $a_i \in S$.
We assume that every element of $R$ is represented with a fixed precision $k \in \Z_{\geq 0}$. This is equivalent to saying that elements are represented as elements of the quotient $R/(\pi^k)$, which is a local principal Artinian ring.
In particular, the algorithms for computing (reduced) resultants developed in Section~\ref{sec:resultant} apply.
Note that this is an important building block, for example in the following tasks:
\begin{itemize}
    \item Computation of the norm of elements in algebraic extensions of $p$-adic fields, which is just a resultant computation.
    \item The computation of a bivariate resultant is one of the main tools to compute the factorization of $p$-adic polynomials, see \cite{FoPaRo2002} and \cite{CantorGordon2000}
\end{itemize}

\subsubsection*{Computing greatest common divisors.}
Apart from these direct applications, the ideas that we used in Section~\ref{sec:resultant} can be exploited also in the computation of the greatest common divisor of polynomials over $R$.
Note that while the classical Euclidean algorithm can be applied in thise setting, it suffers from a great loss of precision, as shown for example in \cite{caruso:hal-01629760}.

Let $f, g\in R[x]$ be two polynomials with coefficient in $R$, for which we want to compute $\gcd(f, g)$. 
First of all, we notice that we can assume that the polynomials are primitive. If not, we can just divide the polynomials by the content.
The division between polynomials, in the case the divisor is monic, works as usual. In the case the divisor $g$ is not monic, we can apply the following adapted version of Proposition~\ref{prop:fun_factor}.
\begin{lemma}\label{lem:coprime_pols}
  Let $f = \sum_{i = 0}^n a_ix^i \in R[x]$ and $g = \sum_{i = 0}^m b_i x^i$ be polynomials over $R$ and assume that $v(a_i) \geq 0$ for $0\leq i \leq n-1$, $v(a_d) = 0$, $v(b_i) > 0$ for $1\leq i \leq m$ and $v(b_0) = 0$.
  Then $f$ and $g$ are coprime.
\end{lemma}
\begin{proof}
  Let $d$ be the greatest common divisor of $f$ and $g$. As $d$ divides $f$ and $g$, the same must hold over the residue ring. However, the greatest common divisor of the projections of $f$ and $g$ over the residue field is constant, because $g$ the valuation of the coefficients of positive degree is greater than zero. As the leading coefficient of $f$ is invertible, the same must hold for $d$. Therefore, $d\in R$ and, since $f$ and $g$ are primitive, $d = 1$.
\end{proof}

\begin{proposition}\label{prop:fun_factor_padic}
Let $f = \sum_{i=0}^d f_i x^i \in R[x]$ be a primitive polynomial. Assume that there exists $0 \leq s \leq d$ such that for $s + 1 \leq i \leq d$ the coefficient $f_i$ has positive valuation and $f_s$ is invertible.
  Then there exists a factorization $f = f_1\cdot f_2$ in $R[x]$ such that $f_1$ modulo $(\pi)$ is constant and $f_2$ is monic of degree $s$.
  The polynomials $f_1$ and $f_2$ are coprime and can be computed using $O(\MM(d) \log k)$ basic operations in $R$.
\end{proposition}

\begin{proof}
The existence of $f_1$ and $f_2$ follow as in the proof of Proposition~\ref{prop:fun_factor}.
The coprimality follows from Lemma~\ref{lem:coprime_pols}
\end{proof}

Therefore, if the leading coefficient of $g$ has positive valuation, we can factorize it as $g = g_1 \cdot g_2$ with $g_1, g_2 \in R[x]$ coprime and $g_2$ monic. In particular, $\gcd(f, g) = \gcd(f, g_1)\gcd(f, g_2)$. 
As $g_2$ is monic, we can divide $f$ by $g_2$ and continue with the algorithm recursively. For what concerns $\gcd(f, g_1)$, we distinguish two cases.
If $f$ is monic, we can conclude that $f$ and $g_1$ are coprime, using again Lemma~\ref{lem:coprime_pols}.
If $f$ is not monic, then we can apply again Proposition~\ref{prop:fun_factor_padic} to $f$ and get a factorization $f = f_1 \cdot f_2$ with $f_1, f_2 \in R[x]$ coprime and $f_2$ monic.
In particular, we have $\gcd(f, g_1) = \gcd(f_1, g_1) \gcd(f_2, g_1)$. As above, $\gcd(f_2, g_1) = 1$, so that $\gcd(f, g_1) = \gcd(f_1, g_1)$.
Therefore we are in the case in which both polynomials are for the form $g_1 = g_{1, 0} + \pi x \tilde g$ and $f_1 = f_{1, 0} + \pi x \tilde f$ for some polynomial $\tilde f, \tilde g \in R[x]$ and $g_{1, 0}$, $f_{1, 0} \in R^\times$. In this case, we again take advantage of the reciprocal polynomials:

\begin{lemma}
 Let $u = \sum_{i = 0}^n a_ix^i \in R[x]$ and $v = \sum_{i = 0}^m b_i x^i$ be polynomials over $R$ and assume that $v(a_i) \geq 0$ for $1\leq i \leq n$, $v(a_0) = 0$, $v(b_i) > 0$ for $1\leq i \leq m$ and $v(b_0) = 0$.
 Then $\gcd(u, v) = i(\gcd(i(u), i(v)))$.
\end{lemma}

\begin{proof}
Let $g\in  R[x]$ be a greatest common divisor of $i(u), i(v)$.
As $i(u), i(v)$ have invertible leading coefficient, the same must hold for $g$ by Gauss's lemma.
Let $s, t \in R[x]$ be such that $i(u) = g\cdot s$ and $i(v) = g \cdot t$. Then $u = i(i(u)) = i(g\cdot s)$ and $v = i(g\cdot t)$. As the constant term of $i(u)$ is non-zero, $i(g)\cdot i(s) = i(g\cdot s)$ and therefore we get $i(g)\mid \gcd(u, v)$. We now need to prove that $i(s)$ and $i(t)$ are coprime, but this follows in the same way by noticing that $s$ and $t$ are monic with non-zero constant term.
\end{proof}

Summarizing, we get the following algorithm:
\begin{algorithm}\label{alg:gcd}
Input: Polynomials $f$, $g \in R[x]$.\\
Output: $\gcd(f, g)$.
\begin{enumerate}
  \item If $f$ or $g$ are not primitive, return $\gcd(c(f), c(g)) \cdot \gcd(f/c(f), g/c(g))$.
  \item If $\deg(g) > \deg(f)$, swap $f$ and $g$.
  \item If $g = 0$, return $f$. If $g$ is a non-zero constant, return $\gcd(c(f), g) \in R$.
  \item If $\lc(g)$ is not invertible, then
  \begin{enumerate}
      \item Split $g$ as $g_1\cdot g_2$ using Proposition \ref{prop:fun_factor_padic}, with $g_2 \in R[x]$ monic.
      \item If $\lc(f)$ is invertible, then return $\gcd(f, \tilde g)$
      \item Split $f$ as $f_1 \cdot f_2$ using Proposition \ref{prop:fun_factor_padic}, with $f_2 \in R[x]$ monic.
      \item Return $\gcd(\tilde f, \tilde g) \cdot i(\gcd(i(f_1), i(g_1)))$.
  \end{enumerate}
  \item Compute the remainder $r$ of the division of $f$ by $g$.
  \item Return $\gcd(g, r)$.
\end{enumerate}
\end{algorithm}

The correctness of the algorithm is clear, as we know that the Euclidean division preserves the greatest common divisor and the factorization of Proposition~\ref{prop:fun_factor_padic} is given by coprime polynomials.
As usual, we can keep track of the transformation we do in order to recover a pair of Bézout coefficients, in the same way as we explained in Section~\ref{subsect:bezout_coeffs}.

\section{Ideal arithmetic in number fields} 
In this section, we apply the algorithms from Section~\ref{sec:resultant} to the basic problems of computing minima and norms of ideals in rings of integers of number fields.
More precisely, let $\gamma$ be an integral primitive element for a number field $K$ and let $f \in \Z[x]$ be its monic minimal polynomial, that is, $K = \Q(\gamma) = \Q[x]/(f)$.
Let $\mathcal O_K$ be the maximal order of $K$ and consider a non-zero ideal of $\mathcal O_K$ with a normal presentation $I = (a, \alpha)$ with $a\in \Z$ and $\alpha\in \mathcal O_K$
as defined in~\cite[6.3 Ideal calculus]{Pohst1997}. 
We will show how to compute the norm and the minimum of $I$ given such a presentation.

\subsection{Norm of an ideal}
 We want to compute the norm $\NN(I) = \lvert \mathcal O_K/I \rvert$ of $I$. Since $I$ has a normal presentation, it is easy to see that $\NN(I) = \gcd(\NN(a), \NN(\alpha)) = \gcd(a^n, \NN(\alpha))$, where $\NN(\alpha) \in \Q$ denotes the usual field norm.

In order to compute the norm of an element, we make use of the following folklore statement:

\begin{lemma}
Let $\alpha\in \Z[\gamma] = \Z[x]/(f)\subseteq K$ and let $g \in \Z[x]$ such that $g(\gamma) = \alpha$. Then $\NN(\alpha) = \res(g, f)$.
\end{lemma}

In particular, the norm of $\alpha$ and therefore the norm of $I$ can be determined using one resultant computation of two polynomials in $\Z[x]$.

On the other hand, we do not need to compute the norm of $\alpha$, but it is sufficient to determine the norm of $\alpha$ modulo $a^n$. 
We can thus compute the norm of an ideal as follows.

\begin{proposition}
Let $I = (a, \alpha)$ and $g \in \Q[x]$ such that $g(\gamma) = \alpha$.
\begin{enumerate}
    \item If $g \in \Z[x]$ and $\bar r = \res(\bar g, \bar f) \in \Z/a^n\Z$ with $\bar g, \bar f \in (\Z/a^n\Z)[x]$, then $\NN(I) = \gcd(a^n, r)$.
    \item Let $d\in \Z$ be an integer such that $d\cdot g \in \Z[x]$ and denote by $g_1 = d\cdot g$. If $\bar r = \res(\bar g_1, \bar f) \in \Z/a^nd^n\Z$ with $\bar g_1, \bar f \in (\Z/a^nd^n\Z)[x]$, then $\NN(I) = \gcd(a^n, \frac{r}{d^n})$.
\end{enumerate}
\end{proposition}

Let us assume first that $\alpha \in \Z[\gamma]$ and let as above $g\in \Z[x]$ such that $g(\gamma) = \alpha$. In order to compute the norm of the ideal, we compute $\res(g, f) \pmod{a^n}$ using Algorithm~\ref{alg:resultant}. If $\alpha$ is not in $\Z[\gamma]$, we compute its denominator $d$ , i.e. the minimum positive integer such that $d\alpha \in \Z[\gamma]$. Write $d = c\cdot u$, where $u$ is the largest divisor of $d$ coprime to $a$. We compute $\res(\tilde {d\alpha}, f) \pmod{c^na^n}$ and return the result divided by $c^n$. The result will be correct as $N(d\alpha) = d^nN(\alpha)$.

Thus we have replaced the resultant computation over $\Z[x]$ with a resultant computation over a residue ring $\Z/m\Z$ for a suitable $m \neq 0$.

\subsection{Minimum of an ideal}
We want to compute the minimum of the ideal $I$, that is, the positive integer $\min(I) \in \Z_{>0}$ which satisfies $\min(I)\Z = I\cap \Z$.
Again we distinguish two cases, depending on the index of $\Z[\gamma]$ in $\mathcal O_K$.
\begin{lemma}
 Assume that $[\mathcal O_K: \Z[\gamma]]$ and $a$ are coprime.
 Then $\min(I) = r$, where $\bar r = \rres(\bar g, \bar f) \in \Z/a\Z$ and $\bar g, \bar f \in (\Z/a\Z)[x]$ and $g(\gamma) = \alpha$.
\end{lemma}
\begin{proof}
First of all, we notice that since $a$ and the index are coprime, the projection of $g$ modulo $a$ is well defined.
Let now $m$ be the minimum of $I$. Then $m$ can be written as an element of $I$, that is, $m = s(\gamma)\cdot a + t(\gamma)\cdot \alpha$ with $s, t \in \Q[x]$ with denominators $d_1$, $d_2$ (the smallest integers $d_1$, $d_2$ such that $d_1\cdot s \in \Z[x]$ and $d_2\cdot t \in \Z[x]$) coprime to $a$.
Thus there exists $u$ such that $m = s(x)\cdot a + t(x)\cdot g(x) + u(x)\cdot f(x)$. Reducing the equation modulo $a$, we get that $m \in (\bar f, \bar g)\cap \Z/a\Z$.
Hence $m$ is in the reduced resultant ideal of $(\bar f, \bar g)$. On the other hand, any relation over $\Z/a\Z$ can be lifted, giving the equality.
\end{proof}

This lemma provides an easy way of computing the minimum of an ideal in the case the index is coprime to the first generator of the ideal. However, it does not work in general, as the denominators of the elements appearing in the equations and the first generator of $I$ do not need to be coprime. In this unlucky case, we take advantage of the following lemma:

\begin{lemma}
  Let $I = (a, \alpha)$ an ideal of $\mathcal O_K$. Then $\min(I)= \gcd(a, \den (\alpha^{-1}))$.
\end{lemma}
This lemma allows us to compute the minimum of $I$ from a Bézout identity. Indeed, let $a\alpha + b f= r$ be the Bézout identity for $f$ and $\alpha$. Then $a/r$ is the inverse of $\alpha$ in the number field. As we want to work modularly, we need to take care of the denominators. Let $d$ be the denominator of $\alpha$ and write $d = d_0 u_0$ where $u_0$ is the largest divisor of $d$ which is coprime to $a$. In the same way, let $e$ be the exponent of $\mathcal O_K/\Z[\gamma]$ and write $e = e_0 v_0$ with $v_0$ the largest divisor coprime to $a$. Then we compute a Bézout identity of $f$ and $d\alpha$ modulo $a\cdot e_0 \cdot d_0$ in order to find the denominator of the inverse. The minimum of the ideal will then be the greatest common divisor of $a$ and the denominator of the modular inverse.

\bibliographystyle{alpha}
\bibliography{resultant}
\end{document}